\documentclass[conference]{IEEEtran}
\IEEEoverridecommandlockouts

\usepackage{cite}
\usepackage{amsmath,amssymb,amsfonts,amsthm}
\usepackage{graphicx}
\usepackage{textcomp}
\usepackage{xcolor}
\usepackage{caption}
\usepackage{algorithm}
\usepackage[noend]{algpseudocode}
\usepackage{xurl}
\usepackage{float}

\def\BibTeX{{\rm B\kern-.05em{\sc i\kern-.025em b}\kern-.08em
    T\kern-.1667em\lower.7ex\hbox{E}\kern-.125emX}}
\makeatletter
\def\fs@ruled{%
  \def\@fs@cfont{\bfseries}%
  \let\@fs@capt\floatc@ruled
  \def\@fs@pre{\vspace*{0.09in}\hrule height .8pt depth0pt \kern2pt}%
  \def\@fs@post{\kern2pt\hrule\relax}%
  \def\@fs@mid{\kern2pt\hrule\kern2pt}%
  \let\@fs@iftopcapt\iftrue
}
\makeatother

\newtheorem{theorem}{Theorem}

\newcommand{\vs}{\boldsymbol{s}}
\newcommand{\vz}{\boldsymbol{z}}
\newcommand{\vx}{\boldsymbol{x}}
\newcommand{\vh}{\boldsymbol{h}}
\newcommand{\vmu}{\boldsymbol{\mu}}
\newcommand{\vell}{\boldsymbol{\ell}}
\newcommand{\mSigma}{\boldsymbol{\Sigma}}

\newcommand{\kaan}[1]{\textcolor{red}{Kaan: #1. }}

\begin{document}

\title{Top-P  Sensor Selection for Target Localization 

\thanks{* indicates equal contribution.}
\thanks{The work was supported by the Army Research Laboratory grant under Cooperative Agreement W911NF-17-2-0196.}
}

\author{\IEEEauthorblockN{Kaan Buyukkalayci$^\dagger$*,  Kyle Pak$^\dagger$*, Merve Karakas$^\dagger$, Xinlin Li$^\dagger$, and Christina Fragouli$^\dagger$\\ 
$^\dagger$University of California, Los Angeles\\
Email:\{kaanbkalayci, whilewak, mervekarakas, xinlinli, christina.fragouli\}@ucla.edu}}

\maketitle

\begin{abstract}
We study set-valued decision rules in which performance is defined by the inclusion of the top-$p$ hypotheses, rather than only the single best or true hypothesis. This criterion is motivated by sensor selection for target tracking, where inexpensive measurements are used to identify a list of sensor nodes that are likely to be closest to a target. We analyze the performance of top-$p$ versus top-$1$ selection under sequential hypothesis testing, propose a geometry-aware sensor selection algorithm, and validate the approach using real testbed data.
\end{abstract}


\section{Introduction}
In a substantial body of information-theoretic work, the decoder or decision rule is permitted to produce a set-valued output rather than a single estimate. This framework includes list decoding in channel coding, where the decoder outputs a list of candidate codewords, as well as list-based hypothesis testing, in which the decision rule returns a subset of hypotheses. In both settings, the fundamental performance criterion is the probability that the true codeword or hypothesis is not contained in the output set.

However, in some applications, the performance criterion is not limited to the inclusion of the true codeword or hypothesis in the output list. Instead, we may have a distance (or similarity) metric on the hypothesis or codeword space and want to evaluate the decoder or decision rule based on whether the output list contains a specified collection of nearest neighbors of the true element (e.g., the first, second, or 
$k$-th closest elements under this metric). In such a setting, error occurs not merely when the true codeword or best hypothesis is missing from the list, but when the list fails to include the metric-defined best-$k$ codewords or hypotheses.

This paper considers such a criterion in the context of sensor selection for target localization. We study a scenario in which a target moves within an area monitored by   sensor nodes and leverage low-cost sensing modalities, such as acoustic measurements, to generate a list of  sensor nodes that are most likely to be closest to the target. The motivation is that higher-fidelity sensing modalities, such as video, are available at the sensor nodes but are not intended to be activated ubiquitously. Instead, we are interested in an algorithm that uses inexpensive measurements to identify a subset of candidate sensors, so that we activate the higher-fidelity modalities only at these  nodes. For our algorithm, performance is not determined solely by the identification of the single closest sensor, but by the quality of the entire selected list: since all selected sensors are activated, the objective is to ensure that each sensor in the list is among the closest to the target to be informative for tracking.

Our contributions in this paper are as follows:
\begin{itemize}
\item  We analytically examine how localization accuracy varies under top-$p$ list selection for different $p$-values within a sequential hypothesis testing framework.
\item  We then propose an algorithm that takes advantage of geometry, such as implicitly the fact that the top $p$ sensors would be close to each other geographically.
We extend this algorithm both for the case of a single and multiple target localization.
\item We experimentally explore how top-$p$ instead of top-$1$ accuracy changes over real testbed data, using data driven modeling of the sensor measurements.
\end{itemize}

\noindent \textbf{Related Work.} Traditional target tracking and localization in WSNs typically use continuous-state estimators (e.g., Kalman filters) and sequential change detection \cite{zhao-2002, page-1954, lorden-1971, tartakovsky-2009}, where sensor activation is then optimized for estimation accuracy or detection delay, with substantial work on energy-aware tracking and coordination to reduce sensing/communication costs \cite{pattem-2003, chen-2003, chen-2006, madhavi-2023}. These approaches are tailored to persistent state reconstruction and generally do not produce ranked or set-valued outputs aligned with downstream actions.

With increasing on-node compute and edge communication in modern deployments, the bottleneck shifts from data acquisition to scalable distributed inference \cite{baek2025,zou2022}. This motivates selective processing and fusion over large sensor collections, including set-valued sensor selection for localization \cite{bhattacharya2021selection}. We take a related but distinct view: rather than estimating coordinates, we seek a small set of sensors likely to be closest to the target, motivated by hybrid sensing architectures where low-cost modalities trigger high-cost actions \cite{kulkarni-2005, jurdak-2010}. Works such as \cite{bsearch,asilomar} similarly produce coarse spatial information via partitioning, and \cite{milcom} proposes a Bayesian-estimation-based single-target method akin to our construction in Section~\ref{sec:algorithm}; however, these do not study ordered set-valued sensor selection as a primary objective.

On the theory side, prior work analyzes consistency/calibration of ordered top-$k$ rules \cite{lapin2016loss, yang2020topk} and adaptive prediction sets that trade list size for error control \cite{sadinle2019least}, with list-decodable learning returning small candidate sets under ambiguity/corruption \cite{karmalkar2019}. These frameworks typically order by labels/posteriors and do not enforce metric-structured correctness, e.g., top-$k$ nearest neighbors under a spatial metric. To our knowledge, the criterion studied here is new in sensor selection/localization.

\section{Setup and Notation}
We write $\|\cdot\|$ to denote the Euclidean distance, and $[i] = \{1,2,\dots,i\}$ for any $i \in \mathbb{N}$ with $i > 0$, where $\mathbb{N}$ denotes the set of natural numbers. For an ordered index set $I = \{i_1, i_2, \dots, i_{|I|}\}$ with $i_1 < i_2 < \dots < i_{|I|}$, we define the vector $\mathbf{v} = [u_i]_{i \in I}$ as $\mathbf{v} := (u_{i_1}, u_{i_2}, \dots, u_{i_{|I|}})^\top$. We further use $\mathbf{1}\{\cdot\}$ to denote the indicator function of an event $A$, equal to $1$ if $A$ holds and $0$ otherwise.

We consider a target whose location $\vell_t$ at time $t$ is unknown along with an area equipped with $s$ sensor nodes whose locations are denoted by $\{\vs_i\}_{i=1}^{s}$, each of which is endowed with multiple sensing modalities exhibiting different information characteristics. The aim is to leverage inexpensive and noisy modalities across all nodes to “recommend” a list of size $p$ consisting of the sensor nodes that are most likely to be the closest $p$ nodes to the target. 

For the sake of analysis in Section~\ref{sec:p_closest analysis} and for the algorithm in Section~\ref{sec:algorithm}, we define a finite set of hypothesis locations $\mathcal{H} = \{\vh_1,\dots\vh_{|\mathcal{H}|}\}$, obtained by discretizing the problem area into an arbitrarily fine square grid, to which the true location $\vell_t$ is assumed to belong. Furthermore, we assume that no two sensor nodes are at the same distance from any location in $\mathcal{H}$ and that target is stationary in each time block indexed by $t$.

\section{Normalized Max Value Selection}
\label{sec:p_closest analysis}

Assuming a parametric propagation model is available, we consider a max-value selection baseline, which, at each time, normalizing the measurements, takes the top-$p$ of these, and returns the associated sensor nodes as a candidate list for the closest-$p$ set. The method is appealing in large deployments due to its minimal computational overhead. In addition, its error probability admits an exact characterization under Gaussian noise and is also reasonable for inexpensive sensing modalities, e.g., acoustic receivers, where standard RSS baselines are commonly approximated by Gaussians that reflect outdoor propagation variability. We next analyze this approach.

\subsection{Modeling of Sensor Measurements}
\label{subsec:model}
In many wireless communication and outdoor acoustic propagation settings, a common baseline model for the received signal strength (RSS) measurements $z_i^t$ (in dB), $i\in[s]$, collected by $s$ sensors from a source at location $\vell_t$ at time $t$, is linear in the log-distance domain.
\[
z^t_i = P_{oi} - \eta_i \cdot 10\log_{10}(d_{i}^t) + \epsilon^t_i,
\]
where $P_{oi}$ and $\eta_i>0$ are, in our case, learnable parameters by ordinary least squares (OLS) regression, $d_i^t=\|\vell_t-\vs_i\|$ denotes the Euclidean distance between the target location and sensor $i$, and $\epsilon_i^t\sim\mathcal N(0,\sigma_i^2)$ is independent zero-mean Gaussian noise corrupting the measurement at sensor $i$. The noise terms $\{\epsilon^t_i\}_{i \in [s]}$ are assumed independent across sensors. Figure~\ref{fig:path_loss_fits} illustrates this fitting on real-world data obtained for different acoustic sensors scattered around the sensing area.

\begin{figure}
    \centering
    \includegraphics[width=0.9\linewidth]{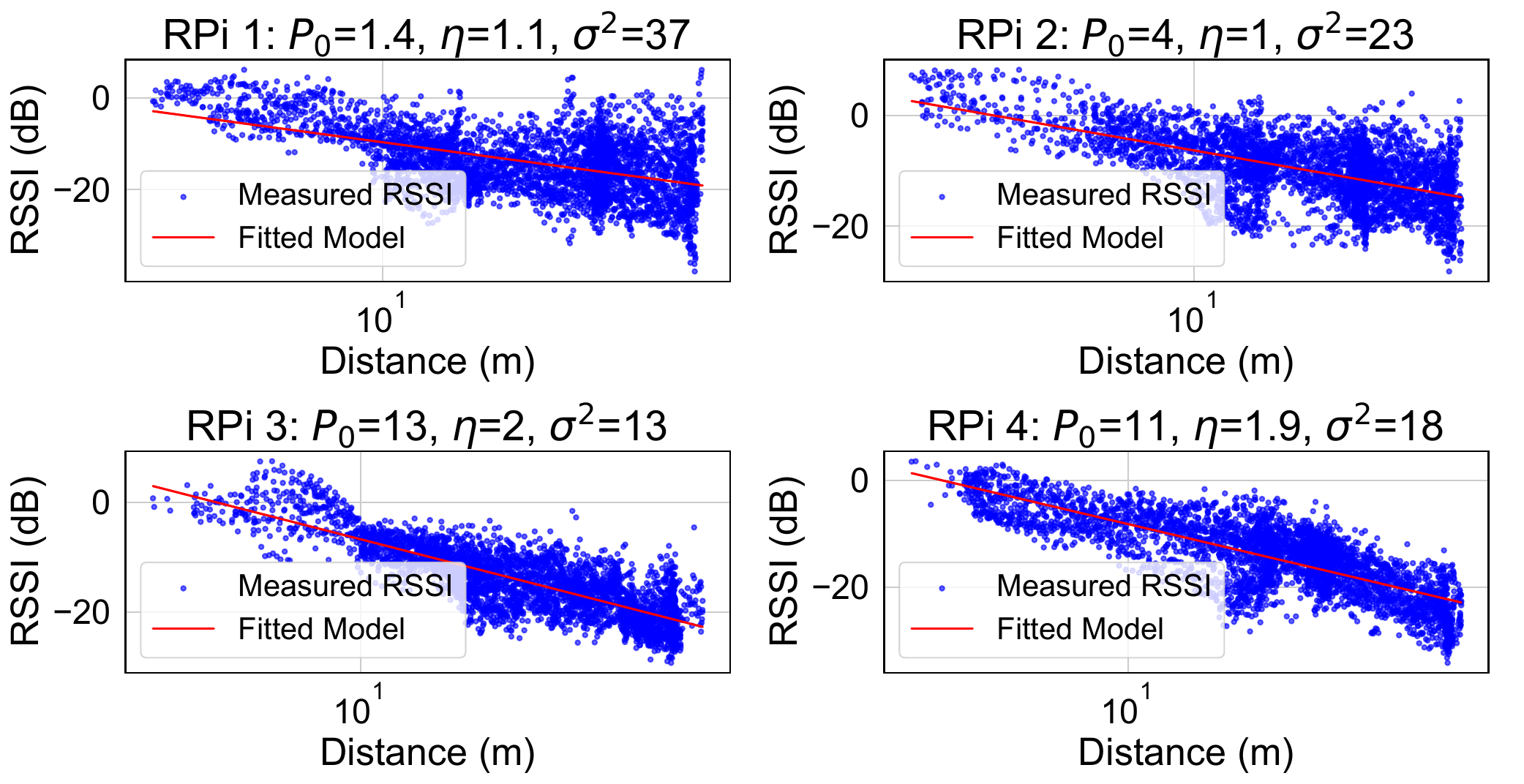}
    \caption{Fitted propagation models for different sensor nodes.}
    \label{fig:path_loss_fits}
    \vspace{-0.25in}
\end{figure}

\subsection{Error Probability under Top-$p$ Selection}
Without loss of generality, the measurements $z_i^t$ can be normalized so that the corresponding observations satisfy
\[
\tilde z_i^t = -10\log_{10}(d_i^t) + \tilde\epsilon_i^t,
\]
where $\tilde\epsilon_i^t \sim \mathcal{N}(0, \tilde\sigma_i^2)$ for $i \in [s]$. This is obtained from the original model by defining $\tilde z_i^t := \big(z_i^t - P_{oi}\big)/\eta_i$ and absorbing the scaling into the noise variance. We drop the superscript $t$ wherever the meaning is clear for the sake of notational clarity and write $\tilde z_i,\tilde\epsilon_i,d_i$ for the remainder of this paper.


\begin{theorem} 
\label{thm:top-p-error}
Consider the top-$p$ decision rule $\widehat K=\arg\max_{|K|=p}\sum_{i\in K} \tilde z_i$. Equivalently $\widehat K$ is the index set of the $p$ largest normalized measurements. Let $K(\vh) \subseteq [s]$ denote the index set of the $p$ closest sensors to $\vh$. Then, under a uniform prior over $\mathcal{H}$ for the true location $\vell_t$, the corresponding error probability, describing the event $\{\widehat K \neq K(\vell_t)\}$, is given by
\[
P_e^{(p)}
=
1-\frac{1}{|\mathcal H|}
\sum_{\vh\in\mathcal H}
\Phi_{p(s-p)}\!\left(\mathbf 0;\ -\mathbf m_{K(\vh)}(\vh),\ \mSigma_{K(\vh)}\right),
\]
where, for an index set $K \subseteq [s]$ and location $\vh \in \mathcal H$,
\[
\mathbf m_{K}(\vh)
:=
\big[\,\mu_i(\vh)-\mu_j(\vh)\,\big]_{(i,j)\in K\times K^c},
\]
with $\mu_i(\vh)=-10\log_{10}\|\vh-\vs_i\|$, $i \in [s]$, and $\mSigma_{K(\vh)}$ has entries indexed by pairs $\left((i,j),(i',j')\right)\in (K(\vh)\times K(\vh)^c) \times (K(\vh)\times K(\vh)^c)$ given by
\begin{align*}
\text{Cov}(\tilde z_i- \tilde z_j,\;\tilde z_{i'}-\tilde z_{j'}) 
=
\tilde \sigma_i^2\mathbf 1\{i=i'\}
+\tilde \sigma_j^2\mathbf 1\{j=j'\}\\
-\tilde \sigma_i^2\mathbf 1\{i=j'\}
-\tilde \sigma_j^2\mathbf 1\{j=i'\}.
\end{align*}
Here $\Phi_{p(s-p)}(\cdot;\mathbf m,\mSigma)$ denotes the CDF of a $p(s-p)$-dimensional Gaussian with mean $\mathbf m$ and covariance $\mSigma$, defined by $\Phi_{p(s-p)}(\mathbf x;\mathbf m,\Sigma)=\int_{(-\infty,\mathbf x]} \mathcal N(\mathbf u;\mathbf m,\Sigma)\,d\mathbf u$.

\end{theorem}

\noindent \textit{Proof Outline.} For all $\vh\in\mathcal{H}$, let $K(\vh)\subseteq[s]$ denote the true top-$p$ index set, with complement $K(\vh)^c=[s]\setminus K(\vh)$. For every $i \in K(\vh)$ and $j \in K(\vh)^c$, requiring that $\tilde z_i - \tilde z_j> 0$ describes the correct event. Since  $\tilde z_i$ and $\tilde z_j$ are Gaussian, $\tilde z_i - \tilde z_j$ is also a Gaussian quantity for all $(i,j)$.

We provide the full proof of Theorem \ref{thm:top-p-error} in Appendix \ref{app:top-p-proof} \cite{Appendix}. Theorem \ref{thm:top-p-error} implies that, although the error probability admits an exact characterization, its numerical evaluation requires computing multiple multivariate Gaussian CDFs, whose computational cost scales poorly as the number of nodes $s$ increases. Moreover, while the expression in Theorem \ref{thm:top-p-error} naturally extends to settings with correlated sensor noise by appropriately modifying the covariance matrices $\mSigma_K(\vh)$, the special structure induced by independent noise across sensors can be exploited to express the error probability in terms of only one-dimensional Gaussian CDFs, leading to a formulation whose computational cost scales more favorably with $s$.

\textbf{Corollary 1.} 
Under the top-$p$ decision rule in Theorem \ref{thm:top-p-error}, the error probability under a uniform prior over $\mathcal H$ admits the following equivalent one-dimensional integral representation:
\begin{align*}
&P_e^{(p)}
=
1-\frac{1}{|\mathcal H|}
\sum_{\vh\in\mathcal H}
\int_{-\infty}^{\infty}
\Bigg[
\prod_{i\in K(\vh)}
\Big(1-\Phi\!\Big(\tfrac{v-\mu_i(\vh)}{\tilde \sigma_i}\Big)\Big)
\Bigg]\\&
\Bigg[
\sum_{j\in K(\vh)^c}
\frac{1}{\tilde \sigma_j}\phi\!\Big(\tfrac{v-\mu_j(\vh)}{\tilde \sigma_j}\Big)
\prod_{\ell\in K(\vh)^c\setminus\{j\}}
\Phi\!\Big(\tfrac{v-\mu_\ell(\vh)}{\tilde \sigma_\ell}\Big)
\Bigg]\, dv.
\end{align*}
where $\phi(\cdot)$ and $\Phi(\cdot)$ denote the probability density function and cumulative distribution function of the standard univariate normal distribution, respectively.

\begin{proof}[Proof Outline]
Letting $U(\vh) \triangleq\min_{i\in K(\vh)} \tilde z_i$ and $V(\vh)\triangleq\max_{j\in K(\vh)^c}\tilde z_j$, the correct event reduces to $
U(\vh)>\; V(\vh)$. Since the measurement noises are independent and Gaussian, the probability $\mathbb P(U(\vh) > V(\vh)\mid \vh)$ can be expressed as a one-dimensional integral by conditioning on $V(\vh)$ and integrating over its marginal distribution.
\end{proof}

\begin{figure}
\centering
\includegraphics[width=0.85\linewidth]{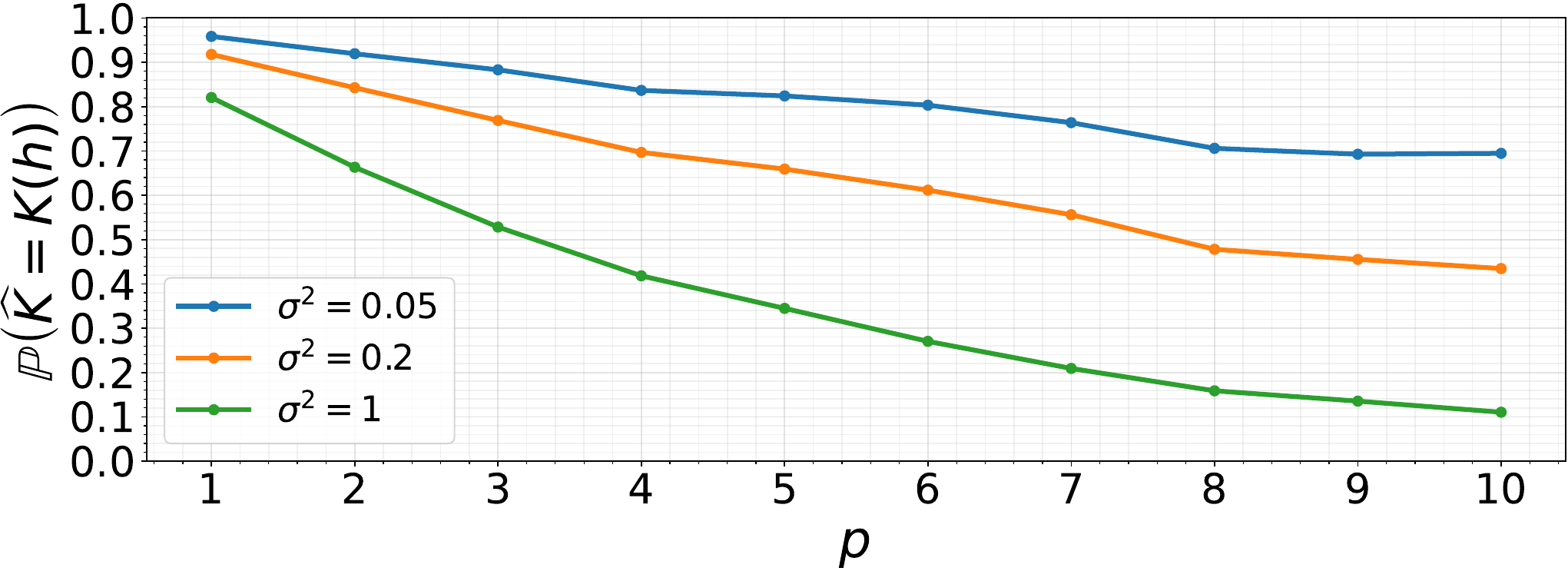}
\caption{Accuracy vs. selection set size for different noise variances and $s$=20 placed randomly in $[0,1]\times[0,1]$}
\vspace{-2em}
\label{fig:acc_vs_p}
\end{figure}

A complete proof is given in Appendix \ref{app:cor}\cite{Appendix}. Numerical evaluation of the expression in Corollary~1 reveals that accuracy decreases rapidly as the output set size $p$ increases. As illustrated in Fig.~\ref{fig:acc_vs_p} for the uniform noise variance case, most of the performance loss occurs as $p$ grows from $1$ to a small number of candidates, where higher variance leads to a steeper initial drop and earlier saturation of accuracy. Beyond a modest value of $p$, accuracy saturates near zero, and further increases in the list size have little additional effect. Consequently, $p$ should be selected as a function of the noise level to meet an application-specific accuracy threshold.

\section{Improved Algorithms and Modeling for List Construction}
\label{sec:algorithm}

Although the decision rule in Section~\ref{sec:p_closest analysis} offers an easily scalable way to construct a list of estimated closest sensor nodes as the number of sensors grows, it does not exploit the inherent spatial structure of the problem, specifically, that measurements are geometrically correlated and that the $p$ closest nodes are themselves likely to be physically close. Instead, a Bayesian estimation–based algorithm may be employed that first computes a posterior distribution over the grid locations $\mathcal H$, and then constructs the output list by selecting sensor nodes that are physically closest to the most probable locations. Although this approach offers improved robustness to noise, its computational cost grows linearly with $|\mathcal H|$ unless special heuristics are applied to approximate the posterior distribution. Sec. \ref{subsec:km-alg} describes such an algorithm.

\subsection{A Bayesian-Estimation Based Algorithm}
\label{subsec:km-alg}

When a model $p(z_i \mid \vh)$ for signal propagation over the sensing environment is available, and a uniform prior over $\mathcal H$ is assumed, the posterior probability of each $\vh \in \mathcal H$ being the true location $\vell_t$ can be computed as
\begin{align*}
\label{eq:bay_update}
&p(\vh\mid\{z_i\}_{i=1}^{s}) \propto p(\{z_i\}_{i=1}^{s} \mid \vh) \, p(\vh) \notag\propto \prod_{i=1}^s p(z_i\mid\vh)
\end{align*}
where $p(\vh) = \mathbb{P}(\vh = \vell_t)$, and the posterior is renormalized to sum to one. Although posterior inference under non-uniform priors is equally straightforward, it is possible to prefer uniform priors in practice to account for potential discontinuities between time blocks in which the algorithm operates.

We form the output set by selecting the $k$ highest-posterior hypotheses $\{\vh^{(1)},\dots,\vh^{(k)}\}$ and, for each $\vh^{(l)}$, taking the indices of the $m$ nearest sensors:
\[
\mathcal M_l
=
\operatorname{BottomM}_{i\in[s]} \ \|\vs_i-\vh^{(l)}\|
\]
where $\operatorname{BottomM}$ returns the $m$ smallest elements of its argument. The output is the union
\(
\widehat{\mathcal S}
=
\bigcup_{l=1}^{k}\mathcal M_l.
\)
In practice, the top posterior hypotheses often cluster, so the sets $\{\mathcal M_l\}_{l=1}^k$ overlap and $|\widehat{\mathcal S}|$ is typically well below the worst-case $km$. 

The full procedure is summarized in Algorithm~\ref{alg:k_m_alg}, where $\operatorname{TopK}$ returns the $k$ largest elements of its argument.


\begin{algorithm}
\caption{Bayesian List Construction via  Aggregation}
\label{alg:k_m_alg}
\begin{algorithmic}[1]
\For{$t = 1,2,\dots$}
    \State Observe measurements $\{z_i\}_{i=1}^s$
    \For{each $\vh \in \mathcal H$}
        \State Compute $p(\vh\mid\{z_i\}_{i=1}^s) \propto \prod_{i=1}^s p(z_i\mid \vh)$
    \EndFor
    \State Select $\{\vh^{(1)},\dots,\vh^{(k)}\} \gets \operatorname{TopK}_{\vh \in \mathcal H}\, p(\vh\mid\{z_i\}_{i=1}^s)$
    \For{$l=1,\dots,k$}
        \State $\mathcal M_l \gets \operatorname{BottomM}_{i\in[s]} \|\vs_i-\vh^{(l)}\|_2$
    \EndFor
    \State Output $\widehat{\mathcal S} \gets \bigcup_{l=1}^k \mathcal M_l$
\EndFor
\end{algorithmic}
\end{algorithm}

For evaluation, let $K(\vell_t)\subseteq[s]$ denote the indices of the true top-$p$ nearest sensors to $\vell_t$. We count time $t$ as successful if $K(\vell_t)\subseteq \widehat{\mathcal S}$, and report empirical accuracy as the fraction of time steps for which this holds.

\subsection{Modeling via Spline Fitting}
\label{subsec:spline}
While the propagation model described in Section \ref{subsec:model} provides a simple way to normalize and compare measurements, one can adopt a linear spline-based modeling approach under a Bayesian estimation setting to capture the heterogeneous and environment-dependent effects observed in practice.

For each sensor node $i$, we partition the range of observed distances into $L$ disjoint bins on a logarithmic scale. Within each bin, we fit a local log-distance model of the form
\[
\mu_{i,w}(d) = P_{o,i,w} - \eta_{i,w}\cdot 10\log_{10}(d),
\]
where $(P_{o,i,w}, \eta_{i,w})$ are bin-specific parameters for bin $w$, estimated from data. Imposing continuity constraints at bin boundaries, one obtains a piecewise-linear spline in the $\log_{10}(d)$ domain. The parameters can then be obtained by solving a constrained least-squares problem that minimizes the squared prediction error over all samples, subject to linear equality constraints enforcing continuity of $\mu_{i,w}(d)$ at adjacent bin endpoints. This results in a convex quadratic program with a unique global minimizer, which can be solved efficiently using standard solvers.

Let $(b_{i,w}, b_{i,w+1})$ denote the distance range of bin $w$ for sensor $i$. For any hypothesis $\vh$ and sensor $i$, we compute the distance $d_i = \|\vh - \vs_i\|$ and identify the corresponding bin $w^\star$ such that $d_i \in (b_{i,w^\star}, b_{i,w^\star+1})$. The likelihood is then evaluated as
\[
p(z_i \mid \vh) = \mathcal N\!\big(z_i;\, \mu_{i,w^\star}(d_i),\, \sigma^2_{i,w^\star}\big),
\]
where $\sigma^2_{i,w^\star}$ is the empirical residual variance estimated within that bin.

Since all spline parameters are learned offline, online likelihood evaluation reduces to a constant-time bin lookup followed by a Gaussian density evaluation. A visualization of representative spline fits is shown in Fig.~\ref{fig:spline_fitting}.

\begin{figure}
\centering
\includegraphics[width=0.8\linewidth]{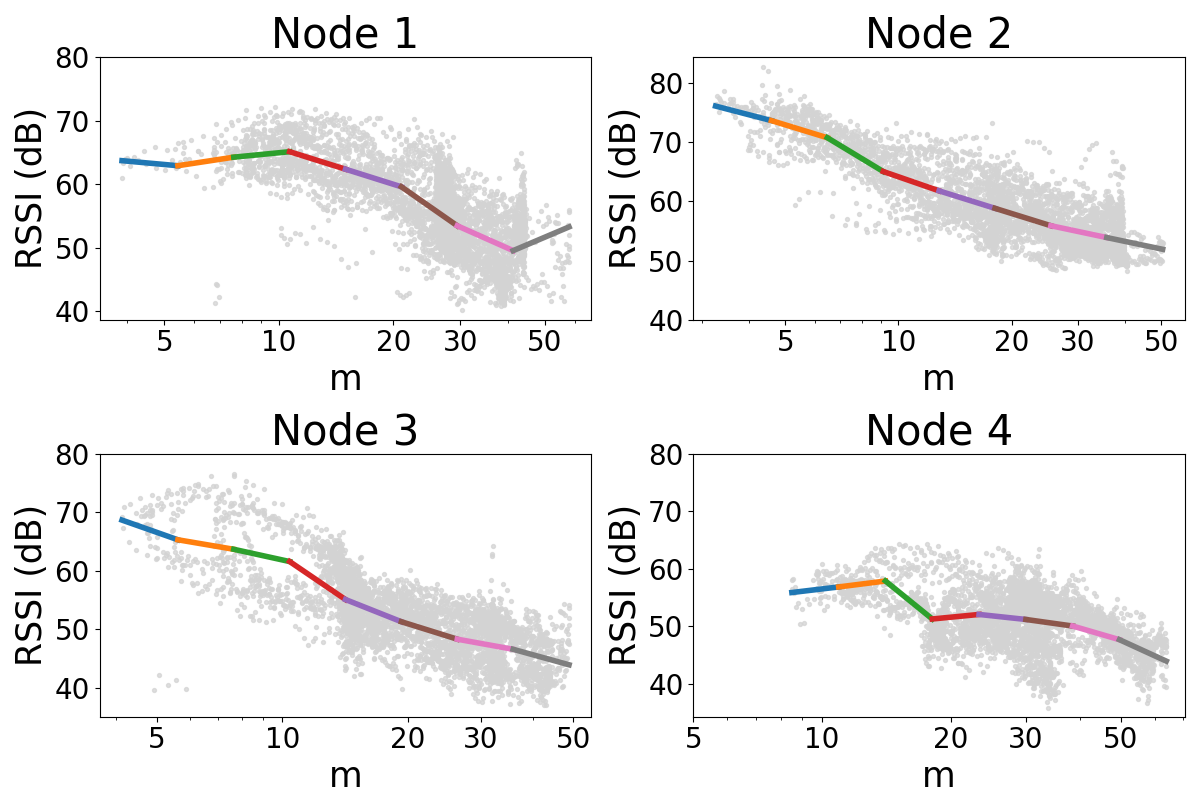}
\caption{Fitted 8-bin linear spline models for various nodes.}

\label{fig:spline_fitting}
    \vspace{-2em}
\end{figure}


\subsection{Localizing Multiple Targets}

A further advantage of the Bayesian-based estimation algorithm described in Section~\ref{subsec:km-alg} is that it is possible to directly extend the procedure to $N\geq 2$ simultaneously moving targets. A naive joint formulation would place the joint state in $\mathcal H^N$, yielding posterior computations that scale as $|\mathcal H|^N$ per time step. To obtain a scalable procedure, we instead restrict the search space by maintaining, for each target $r\in[N]$, a time-varying \emph{local} hypothesis set $\mathcal H_r(t)\subseteq\mathcal H$.

At synchronization times $t\in\{t_{\mathrm{sync}},2t_{\mathrm{sync}},\ldots\}$, we assume access to the true locations $\{\vell_t^{(r)}\}_{r=1}^N$ (e.g., from a higher-fidelity modality). We initialize each $\mathcal H_r(t)$ as a square subgrid of $\mathcal H$ centered at $\vell_t^{(r)}$ with a prescribed initial side length. Between synchronization times, we keep the center fixed and expand the region in discrete steps to account for target motion (equivalently, by adding an outer ``ring'' of grid points), where the step size is defined by the upper bound of the target's velocity. Thus, at time $t$ we perform inference only over the reduced joint 
\vspace{-0.1em}
\[
\mathcal H_1(t)\times\cdots\times\mathcal H_N(t),
\]
whose size is typically far smaller than $|\mathcal H|^N$.

Let $\mathbf{\bar{h}} \triangleq (\vh^{(1)},\dots,\vh^{(N)})$ denote a joint hypothesis with $\vh^{(r)}\in\mathcal H_r(t)$. Given a uniform prior over $\mathcal H_1(t)\times\cdots\times\mathcal H_N(t)$, the posterior satisfies
\vspace{-0.5em}
\[
p(\mathbf{\bar{h}}\mid\{z_i\}_{i=1}^{s}) \propto \prod_{i=1}^s p(z_i\mid \mathbf{\bar{h}}),
\]

where we model multi-target superposition in the scalar domain. Letting $\mu_{i,r}(\vh)$ denote the learned mean RSS in dB during modeling at sensor $i$ for target $r$ at location $\vh$, and defining the corresponding scalar-domain mean RSS $P_{i,r}(\vh)\triangleq 10^{\mu_{i,r}(\vh)/10}$, the measurement $z_i$ at sensor $i$ is commonly modeled to be generated by,

\[
z_i = 10\log_{10}\!\left(\sum_{r=1}^N  P_{i,r}\!\big(\vh^{(r)})\right) + \epsilon_i,
\]

where $\epsilon_i\sim\mathcal N(0,\sigma_i^2)$ is the zero-mean independent Gaussian noise attributed to the sensor.

We select the top-$k$ joint hypotheses under this posterior, denoted by $\text{TopK}$. For each selected joint hypothesis $\mathbf{\bar{h}}$, and for each target component $\vh^{(r)}$, we compute the index set of the $m$ closest sensors to $\vh^{(r)}$:
\[
\mathcal M_r(\mathbf{\bar{h}})
=
\operatorname{BottomM}_{i\in[s]} \ \|\vs_i-\vh^{(r)}\|.
\]

The final recommended set is the union over the selected joint hypotheses and all targets,
\[
\widehat{\mathcal S}
=
\bigcup_{r=1}^N \ \bigcup_{\mathbf{\bar{h}}\in \text{TopK}} \mathcal M_r(\mathbf{\bar{h}}),
\]
yielding a single set intended to cover sensors near all targets.

For evaluation, let $K^{(r)}(\vell_t^{(r)})\subseteq[s]$ denote the indices of the true top-$p$ nearest sensors to target $r$ at time $t$. We declare success at time $t$ if
$
\bigcup_{r=1}^N K^{(r)}(\vell_t^{(r)}) \ \subseteq\ \widehat{\mathcal S},
$
and report empirical accuracy as the fraction of successful time steps. 

The resulting procedure is summarized in Algorithm~\ref{alg:multi_km}.

\begin{algorithm}
\caption{Multi-Target Bayesian List Construction with Synchronized Local Grids}
\label{alg:multi_km}
\begin{algorithmic}[1]
\State Initialize centers $c_r\gets \text{null} ~\forall r\in[N]$ and counter $q\gets 0$
\For{$t=1,2,\dots$}
    \State Observe measurements $\{z_i\}_{i=1}^s$
    \If{$t \bmod t_{\mathrm{sync}} = 0$}
        \State Obtain true locations $\{\vell_t^{(r)}\}_{r=1}^N$, set $q \gets 0$
        \For{$r=1,\dots,N$}
            \State $c_r \gets \vell_t^{(r)}$
        \EndFor
    \Else
        \State $q \gets q+1$
    \EndIf

    \For{$r=1,\dots,N$}
        \State $\mathcal H_r(t)\gets$ square subgrid with center $c_r$, side $B_q$
    \EndFor

    \State Compute $p(\bar{\mathbf h}\mid \{z_i\}_{i=1}^s)$ over $\bar{\mathbf h}\in \mathcal H_1(t)\times\cdots\times\mathcal H_N(t)$
    \State $\widehat{\mathcal H}_{\mathrm{top}}(t)\gets \operatorname{TopK}\!\big(p(\bar{\mathbf h}\mid \{z_i\}_{i=1}^s)\big)$
    \State $\widehat{\mathcal S}\gets \emptyset$
    \For{each $\bar{\mathbf h}=(\vh^{(1)},\dots,\vh^{(N)})\in \widehat{\mathcal H}_{\mathrm{top}}(t)$}
        \For{$r=1,\dots,N$}
            \State $\mathcal M_r(\bar{\mathbf h}) \gets \operatorname{BottomM}_{i\in[s]}\|\vs_i-\vh^{(r)}\|$
            \State $\widehat{\mathcal S}\gets \widehat{\mathcal S}\cup \mathcal M_r(\bar{\mathbf h})$
        \EndFor
    \EndFor
    \State Output $\widehat{\mathcal S}$
\EndFor
\end{algorithmic}
\end{algorithm}

\section{Experimental Results}
\label{sec:exp}
In this section, we evaluate the proposed algorithms using real-world acoustic data collected from outdoor experiments involving both single and multiple moving vehicles. The data were obtained from a sensing field of approximately $10{,}000~\text{m}^2$ containing several buildings and minimal structured background interference, with ambient noise primarily consisting of wind and natural outdoor sounds.

In each experiment, one or two ATVs equipped with GPS transmitters followed the predefined trajectory (Fig.~\ref{fig:gq_map_scrambled}) at an average speed of approximately $3~\text{m/s}$
 for $20$-$30$ minutes. Ten Raspberry Pi nodes were deployed across the field, each with a single-channel USB microphone transmitting received signal strength every $200\,\mathrm{ms}$ 
 synchronized with GPS measurements. Implementation details of the proposed algorithms are provided in Appendix~\ref{app:exps} \cite{Appendix}. Performance is measured in terms of accuracy, defined as the fraction of trials in which the output set contains the $p$ closest sensor nodes to (each) target.

\subsection{Single Vehicle Experiment}

\begin{figure}[t]
    \centering
    \includegraphics[width=0.7\linewidth]{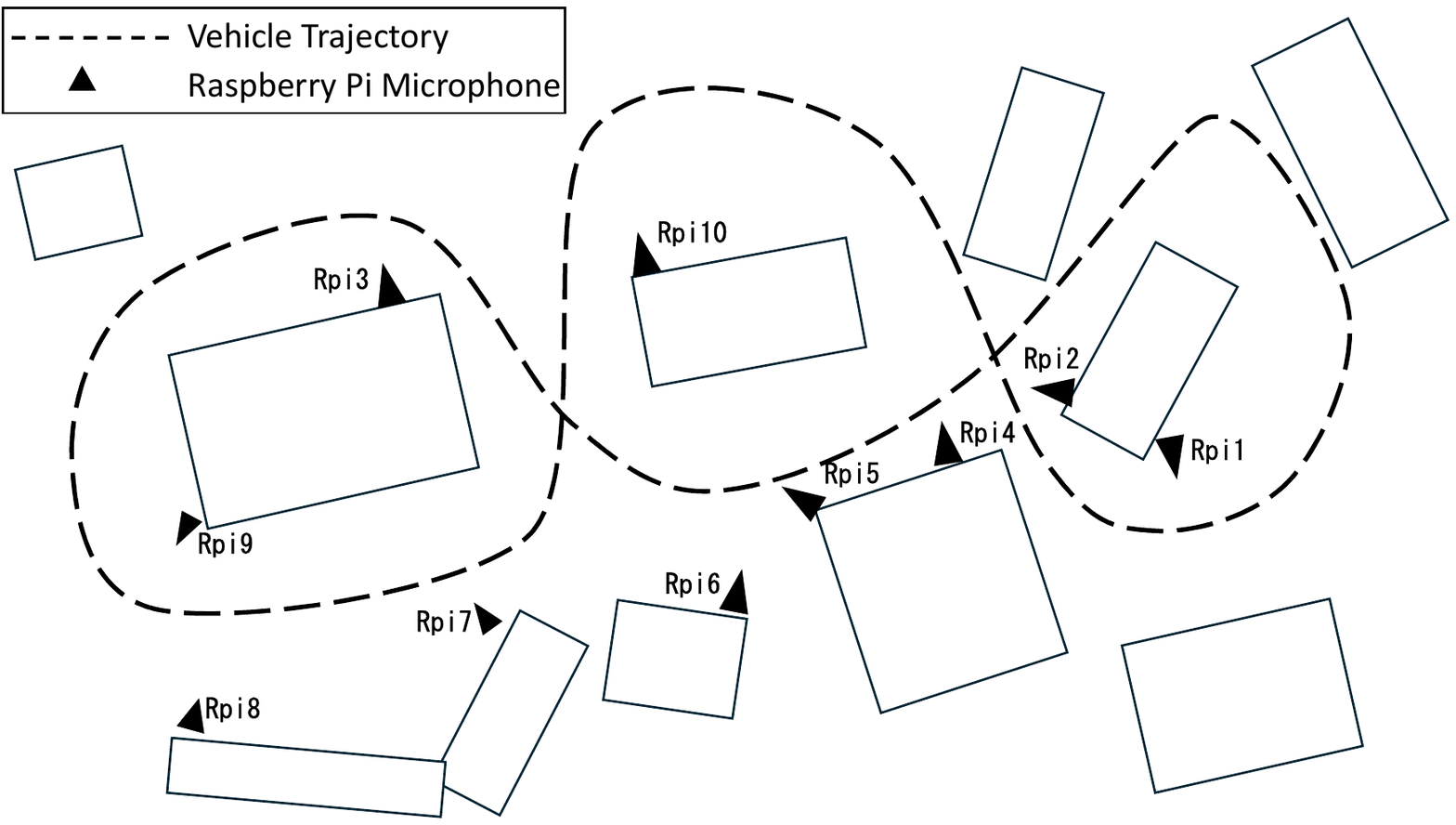}
    \caption{Vehicle trajectory and sensor node placements}
    \label{fig:gq_map_scrambled}
    \vspace{-2em}
\end{figure}

We report performance as a function of the parameter triplet $(k,m,p)$, where $(k,m)$ are the list-construction parameters of Algorithm~\ref{alg:k_m_alg}, specifying the number of top-likelihood hypotheses considered and the number of nearest sensors selected per hypothesis, respectively, and $p$ denotes the number of true nearest sensors that must be contained in the output set for the trial to be counted as correct.


As shown in Fig.~\ref{fig:single}, 
accuracy increases with  $m$ for both methods, while Algorithm~\ref{alg:k_m_alg} 
consistently outperforms the normalized max-value selection baseline described in Section~\ref{sec:p_closest analysis}, especially at larger $p$.
In this experimental setup, we fix $k=1$, so the average output set size of Algorithm~\ref{alg:k_m_alg} is $m$, while the normalized max-value baseline is allowed to output the $m$ largest normalized measurements (rather than $p$). Fig.~\ref{fig:single_set} further illustrates the relationship between accuracy and the average output set size for different values of $p$ under Algorithm \ref{alg:k_m_alg}. We observe a reliable, monotonic increase in accuracy as the output set size grows, highlighting the fundamental tradeoff between set size and top-$p$ containment performance.


\begin{figure}
\centering
\includegraphics[width=0.85\linewidth]{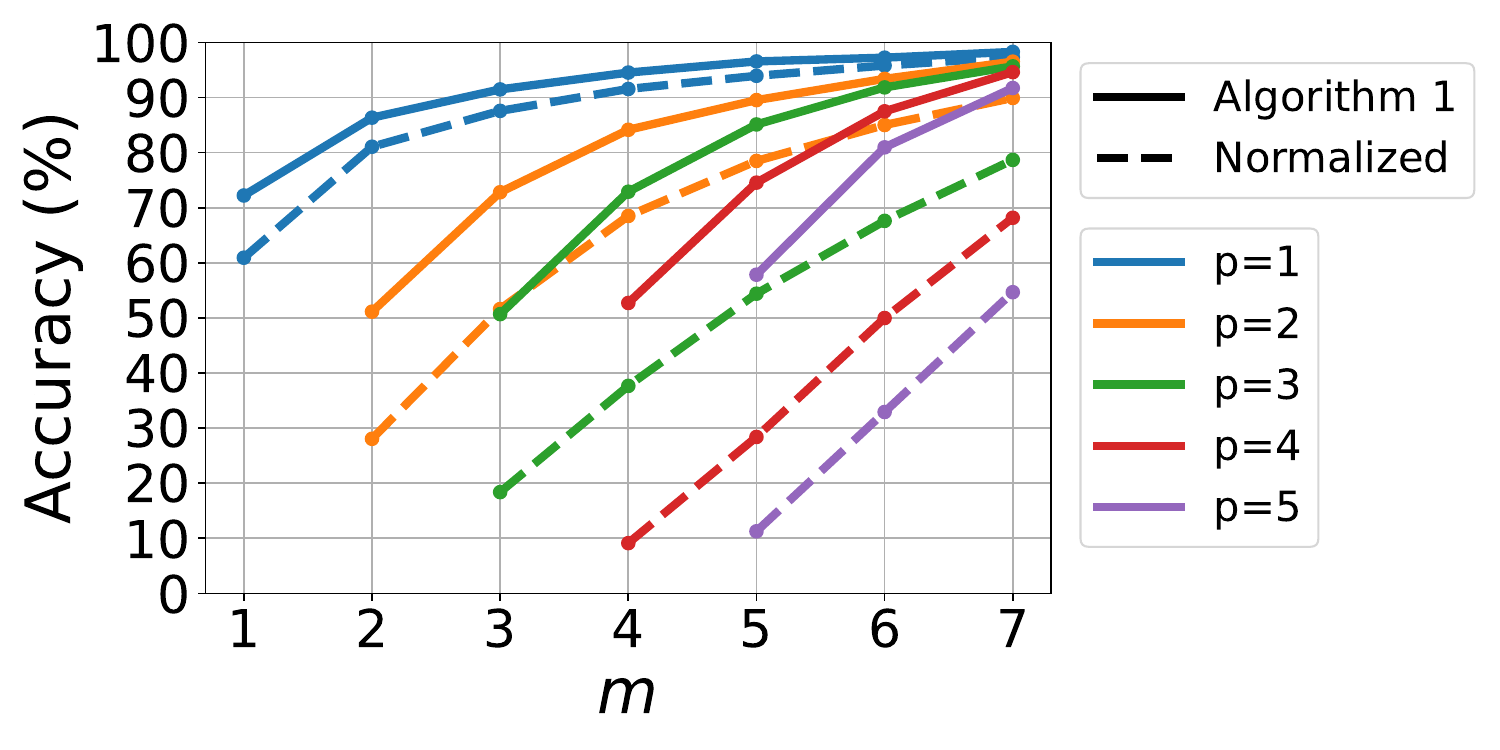}
\captionof{figure}{Algorithm 1 and the Normalized Max Value Selection described in Section \ref{sec:p_closest analysis} for different values of $p$ where $k=1$} 
\label{fig:single}
\end{figure}

\begin{figure}
\centering
\includegraphics[width=0.6\linewidth]{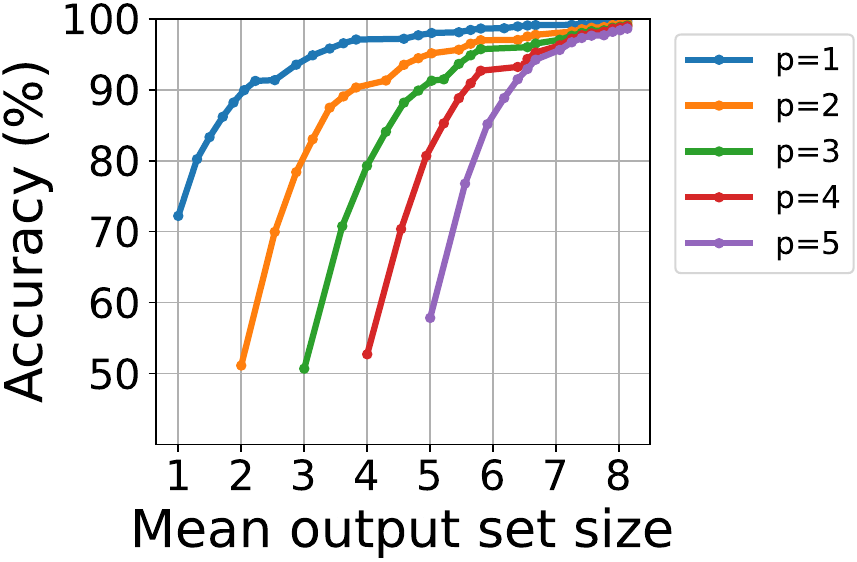}
\captionof{figure}{Accuracy vs average output set size for Algorithm \ref{alg:k_m_alg}} 
\label{fig:single_set}
\vspace{-1.7em}
\end{figure}

\subsection{Multiple Vehicle Experiment}

We next evaluate the multi-target list construction algorithm described in Algorithm \ref{alg:multi_km} in a two-vehicle setting, where two ATVs traverse the sensing area simultaneously.

Fig.~\ref{fig:multiple} reports the top-$p$ containment accuracy as a function of the synchronization interval $t_\text{sync}$ for different values of $p$, with fixed parameters $k=3$ and $m=5$. We observe that, as expected, accuracy degrades as the synchronization interval increases, since longer intervals lead to larger local grids and consequently greater uncertainty in the posterior.
\vspace{-2em}

\begin{figure}
\centering
\includegraphics[width=0.75\linewidth]{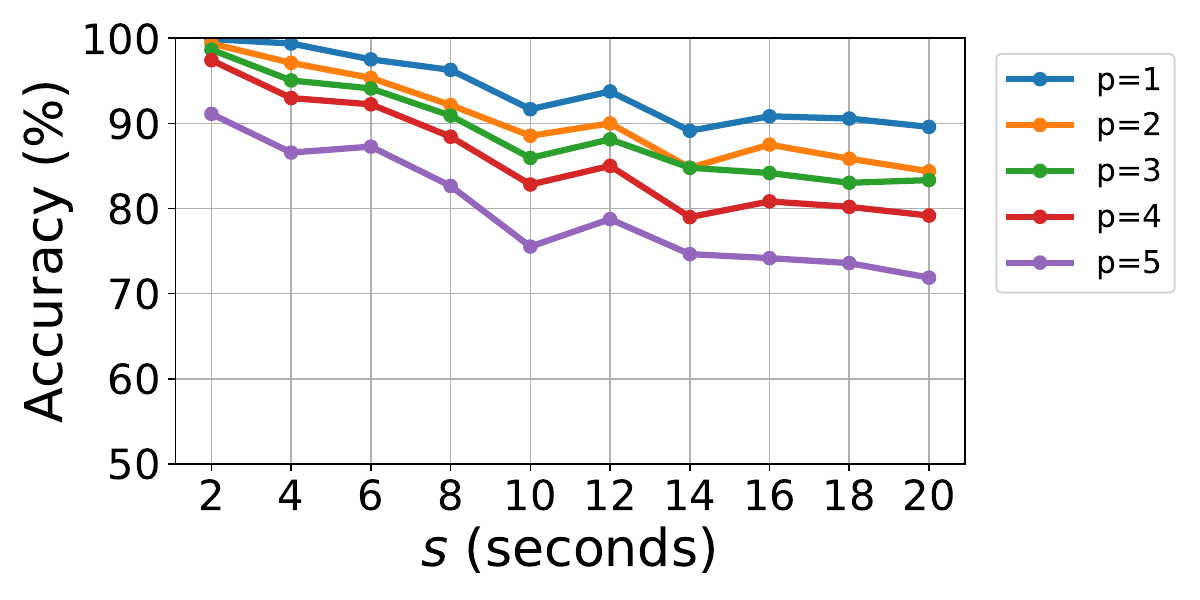}
\captionof{figure}{Accuracy vs synchronization interval $t_\text{sync}$ for different $p$ values ($k=3$, $m=5$), with average output set size $8$} 
\label{fig:multiple}
\end{figure}


\newpage

\IEEEtriggeratref{11}
\bibliographystyle{IEEEtran}
\bibliography{reference}

\onecolumn 
\appendices 

\section{Proof of Theorem \ref{thm:top-p-error}}

\label{app:top-p-proof}

For all $\vh\in\mathcal{H}$, let $\vmu(\vh)\triangleq[\mu_1(\vh),\ldots,\mu_s(\vh)]^\top$, and let $K(\vh)\subseteq[s]$ denote the true unordered top-$p$ index set, with complement $K(\vh)^c=[s]\setminus K(\vh)$. Define the $p(s-p)$-dimensional difference vector
\[
\mathbf D_{K(\vh)}(\vh)
\triangleq
\big[\, \tilde z_i- \tilde z_j \,\big]_{(i,j)\in K(\vh)\times K(\vh)^c}
\in\mathbb R^{p(s-p)}.
\]
Then $\widehat K=K(\vh)$ if and only if $\mathbf D_{K(\vh)}(\vh)\succeq \mathbf 0$. Letting $\vz = [\tilde z_1,\ldots,\tilde z_s]^\top$, we can write $\mathbf D_{K(\vh)}(\vh)=A_{K(\vh)}\vz$, where $A_{K(\vh)} \in \mathbb{R}^{p(s-p)\times s}$ is the matrix whose row indexed by a pair $(i,j)\in K(\vh)\times K(\vh)^c$ equals $e_i^\top - e_j^\top$, and hence

\[
\mathbf D_{K(\vh)}(\vh)\mid \vh
\sim
\mathcal N\!\big(\, \mathbf m_{K(\vh)}(\vh),\, \mSigma_{K(\vh)} \big),
\]
where $\mathbf m_{K(\vh)}(\vh)\triangleq A_{K(\vh)}\vmu(\vh)$ and $\mSigma_{K(\vh)}\triangleq A_{K(\vh)}\mSigma_s A_{K(\vh)}^\top$ with $\mSigma_s = \text{diag}(\tilde \sigma^2_1,\dots,\tilde\sigma_s^2)$. This corresponds to, for $(i,j),(i',j')\in K(\vh)\times K(\vh)^c$,
\begin{align*}
&\big(\mathbf m_{K(\vh)}(\vh)\big)_{(i,j)}=\mu_i(\vh)-\mu_j(\vh)
\end{align*}
and,
\begin{align*}
&\text{Cov}(\tilde z_i-\tilde z_j,\;\tilde z_{i'}-\tilde z_{j'})
=
\tilde \sigma_i^2\mathbf 1\{i=i'\}
+\tilde\sigma_j^2\mathbf 1\{j=j'\}- \tilde\sigma_i^2\mathbf 1\{i=j'\}
-\tilde \sigma_j^2\mathbf 1\{j=i'\}.
\end{align*}
Therefore, the conditional correct-classification probability is given by the orthant probability of a multivariate Gaussian,
\begin{align*}
\mathbb P(\widehat K=K(\vh)\mid \vh)
&=\mathbb P(\mathbf D_{K(\vh)}(\vh)\succeq \mathbf 0\mid \vh) = \mathbb P(-\mathbf D_{K(\vh)}(\vh)\preceq \mathbf 0\mid \vh)=\Phi_{p(s-p)}\!\left(\mathbf 0;\ -\mathbf m_{K(\vh)}(\vh),\ \mSigma_{K(\vh)}\right).
\end{align*}
Finally, under a uniform prior on $\vell_t$ over the square grid $\mathcal H$,
\begin{align*}
P_e^{(p)}
&=1-\frac{1}{|\mathcal H|}\sum_{\vh\in\mathcal H}\mathbb P(\widehat K=K(\vh)\mid \vh)=1-\frac{1}{|\mathcal H|}\sum_{\vh\in\mathcal H}
\Phi_{p(s-p)}\!\left(\mathbf 0;\ -\mathbf m_{K(\vh)}(\vh),\ \mSigma_{K(\vh)}\right).
\end{align*}

\section{Proof of Corollary 1}
\label{app:cor}

For any fixed $\vh\in\mathcal H$, the event \{$\widehat K=K(\vh)$\} is equivalent to the system of inequalities $\tilde z_i\ge \tilde z_j$ for all $i\in K(\vh)$ and $j\in K^c$, which is equivalent to
\[
\min_{i\in K(\vh)} \tilde z_i \;>\; \max_{j\in K(\vh)^c} \tilde z_j,
\]
since ties occur with probability zero under a continuous distribution. Define
\[
U(\vh)\triangleq \min_{i\in K(\vh)} \tilde z_i,\qquad V(\vh)\triangleq \max_{j\in K(\vh)^c} \tilde z_j.
\]
Then $\mathbb P(\widehat K=K(\vh)\mid \vh)=\mathbb P(U(\vh)>V(\vh)\mid \vh)$, and by conditioning on $V(\vh)$,
\[
\mathbb P(U(\vh)>V(\vh)\mid \vh)=\int_{-\infty}^{\infty}\mathbb P(U(\vh)>v\mid \vh)\, f_{V(\vh)}(v\mid \vh)\, dv.
\]
For the first term, by independence within $\tilde z_i$,
\begin{align*}
\mathbb P(U(\vh)>v\mid \vh)
&=
\prod_{i\in K(\vh)}\mathbb P(\tilde z_i>v\mid \vh)=
\prod_{i\in K(\vh)}\Big(1-\Phi\!\Big(\tfrac{v-\mu_i(\vh)}{\tilde \sigma_i}\Big)\Big).
\end{align*}
For the second term, since $V(\vh)=\max_{j\in K(\vh)^c} \tilde z_j$ and the variables $\{\tilde z_j:j\in K(\vh)^c\}$ are independent,
\[
\mathbb P(V(\vh)\le v\mid \vh)
=
\prod_{j\in K(\vh)^c}\Phi\!\Big(\tfrac{v-\mu_j(\vh)}{\tilde \sigma_j}\Big).
\]
Differentiating the above with respect to $v$ gives the density
\[
f_{V(\vh)}(v\mid \vh)
=
\sum_{j\in K(\vh)^c}
\frac{1}{\tilde \sigma_j}\phi\!\Big(\tfrac{v-\mu_j(\vh)}{\tilde\sigma_j}\Big)
\prod_{\ell\in K(\vh)^c\setminus\{j\}}
\Phi\!\Big(\tfrac{v-\mu_\ell(\vh)}{\tilde\sigma_\ell}\Big).
\]

Hence, the probability of the correct event for all $\vh \in \mathcal H$ is,

\begin{align*}
\mathbb P(\widehat K=K(\vh)\mid \vh)
=\int_{-\infty}^{\infty}
\Bigg[
\prod_{i\in K(\vh)}
\Big(1-\Phi\!\Big(\tfrac{v-\mu_i(\vh)}{\tilde \sigma_i}\Big)\Big)
\Bigg]
\Bigg[
\sum_{j\in K(\vh)^c}
\frac{1}{\tilde\sigma_j}\phi\!\Big(\tfrac{v-\mu_j(\vh)}{\tilde\sigma_j}\Big)
\prod_{\ell\in K(\vh)^c\setminus\{j\}}
\Phi\!\Big(\tfrac{v-\mu_\ell(\vh)}{\tilde\sigma_\ell}\Big)
\Bigg]\, dv,
\end{align*}

\clearpage
\section{Details on Experimental Results}

\label{app:exps}

\subsection{Fitted Models for All Sensor Nodes Used in Section~\ref{sec:exp}}

Figure~\ref{fig:pathloss10} shows the fitted propagation models for all sensor nodes used in the experiments of Section~\ref{sec:exp}, based on the log-linear path loss model described in Section~\ref{subsec:model}. Similarly, Figure~\ref{fig:spline10} presents the fitted models for all sensor nodes using the spline-based approach described in Section~\ref{subsec:spline}. 

\begin{figure}[h]
\centering
\begin{minipage}{0.48\linewidth}
    \centering
    \includegraphics[width=\linewidth]{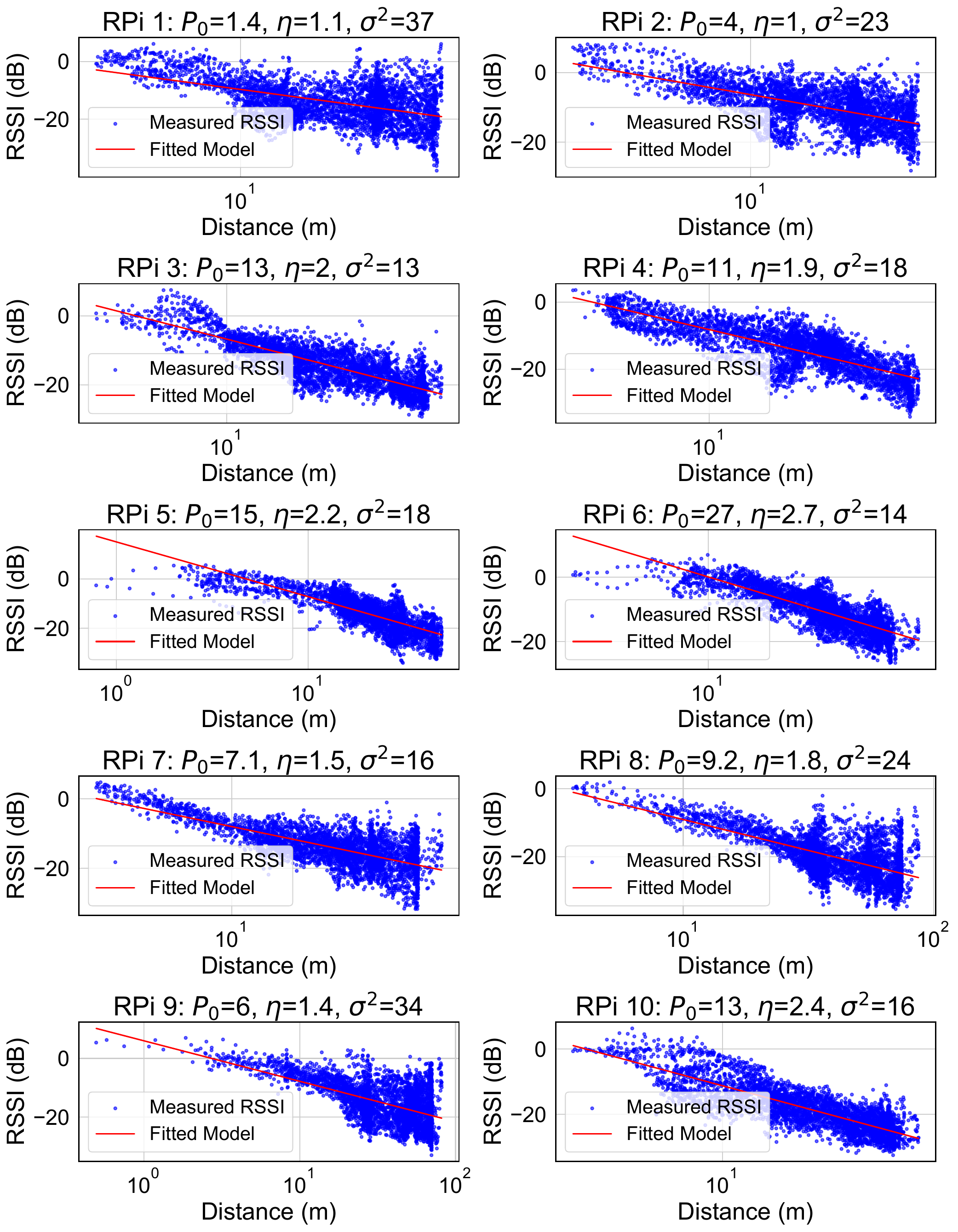}
    \captionof{figure}{Fitted propagation models for all 10 nodes}
    \label{fig:pathloss10}
\end{minipage}\hfill
\begin{minipage}{0.52\linewidth}
    \centering
    \includegraphics[width=\linewidth]{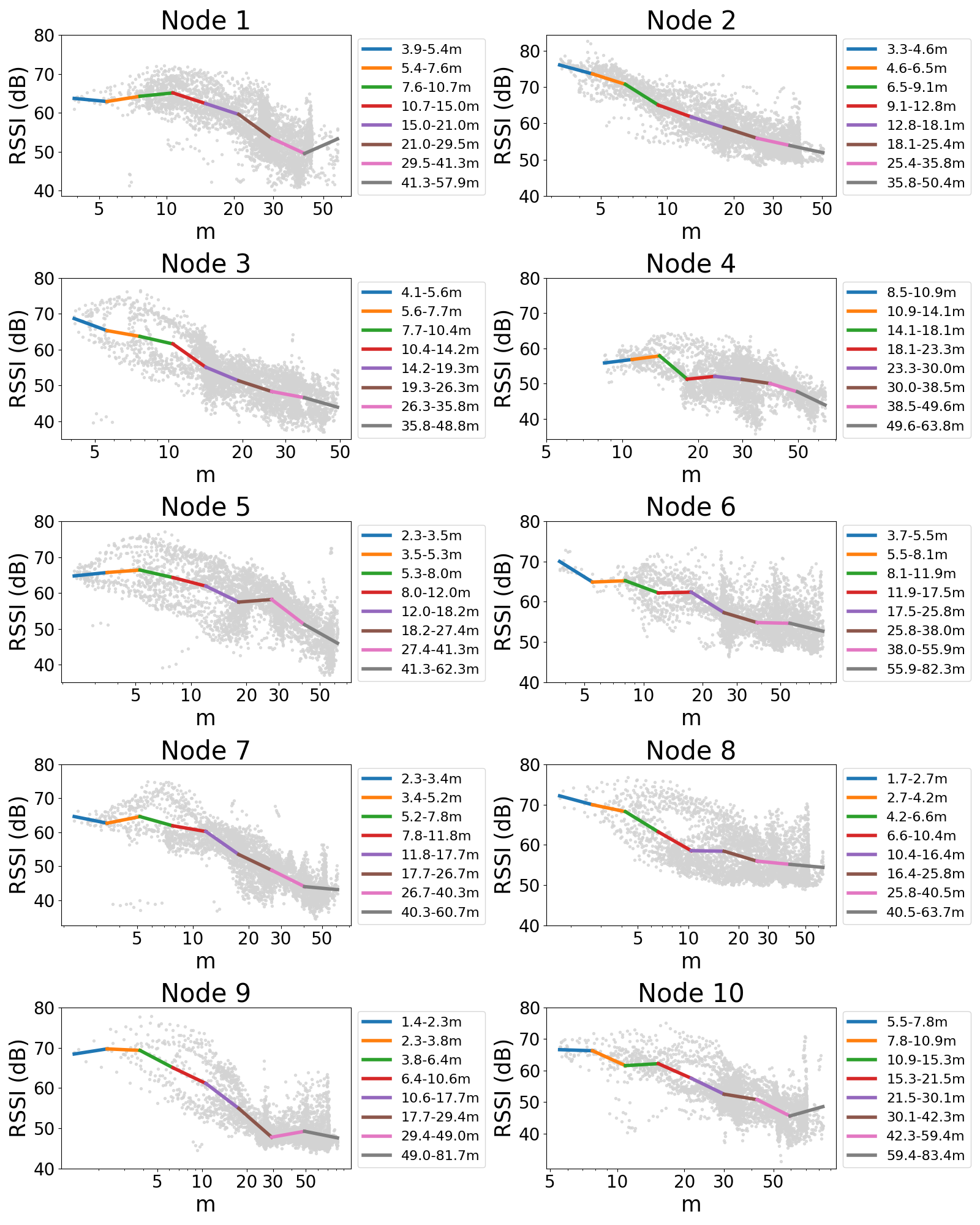}
    \captionof{figure}{Fitted 8-bin linear spline models for all 10 nodes}
    \label{fig:spline10}
\end{minipage}
\end{figure}

\subsection{Single Vehicle Experiment}
Table~\ref{tab:single_params} summarizes additional experimental parameters used in the single-vehicle experiment of Section~\ref{sec:exp} and in Figures~\ref{fig:single} and~\ref{fig:single_set}. The spatial hypothesis set $\mathcal{H}$ is constructed as a $20\times20$ grid with uniform spacing of $7\,\mathrm{m}$ between adjacent locations, yielding $|\mathcal{H}|=400$ hypotheses. The propagation model is discretized into $L=8$ bins and used to construct the fitted likelihood functions described in Section~\ref{subsec:spline}. In the experiments of Fig.~\ref{fig:single_set}, the list construction parameters $k$ and $m$ are varied over $\{1,\dots,7\}$, and the output set size is recorded after each list construction at each time $t$.

\begin{table}[h!]

\centering
\caption{Additional Parameters utilized in Fig. \ref{fig:single}}
\renewcommand{\arraystretch}{1.3}
\begin{tabular}{|c|c|l|}
\hline
\textbf{Parameter} & \textbf{Value(s)} & \textbf{Description}                             \\ \hline
\(L\) & 8 & Number of bins. \\ \hline
Hypothesis spacing      & 7 m             & Spacing in between each hypothesis location.   \\ \hline
\(\mathcal{H}\)   & $20^2=400$      & Number of hypotheses.                \\ \hline
\end{tabular}
\vspace{0.3cm}
\label{tab:single_params}
\end{table}

\newpage

\subsection{Multiple Vehicle Experiment}

The multi-vehicle experiments follow the procedure described in Algorithm~\ref{alg:multi_km}. The additional parameters governing spline modeling, hypothesis grid construction, and local grid evolution are summarized in Table~\ref{tab:table3}. Between synchronization times, let $q\in\{0,1,\dots,t_{\mathrm{sync}}-1\}$ denote the elapsed time index and define the grid side length as $B_q := B_{\lfloor q / t_{\mathrm{inc}} \rfloor}$. Figure~\ref{fig:multi_size} shows the corresponding average output set size for the experiments reported in Figure~\ref{fig:multiple}.

\begin{table}[h!]
\centering
\caption{Additional Parameters utilized in Fig. \ref{fig:multiple}}
\renewcommand{\arraystretch}{1.3}
\begin{tabular}{|c|c|l|}
\hline
\textbf{Parameter} & \textbf{Value(s)} & \textbf{Description}                                   \\ \hline
\(L\) & 8 & Number of bins. \\ \hline
Hypothesis spacing      & 6 m             & Spacing in between each hypothesis location.   \\ \hline
\(\mathcal{H}\)   & $27^2=729$       & Number of hypotheses.                \\ \hline
\(q\) & 0, 1, 2, ..., \(t_{sync}-1\) & Time indices in interval of $t_{sync}$ seconds long.\\ \hline
\(t_{inc}\)       & 2 s                    & Time for next hypothesis grid increase.  \\ \hline
\(B_q\)       & 3, 5, 7, ..., 729 hypothesis locations & Square grid side length for \(q = 0, t_{inc}, 2t_{inc}, 3t_{inc}, ... = 0, 2\text{ s}, 4\text{ s}, 6\text{ s}, ...\) \\ \hline

\end{tabular}
\vspace{0.3cm}
\label{tab:table3}
\end{table}

\begin{figure}[h!]
\centering
\includegraphics[width=0.6\linewidth]{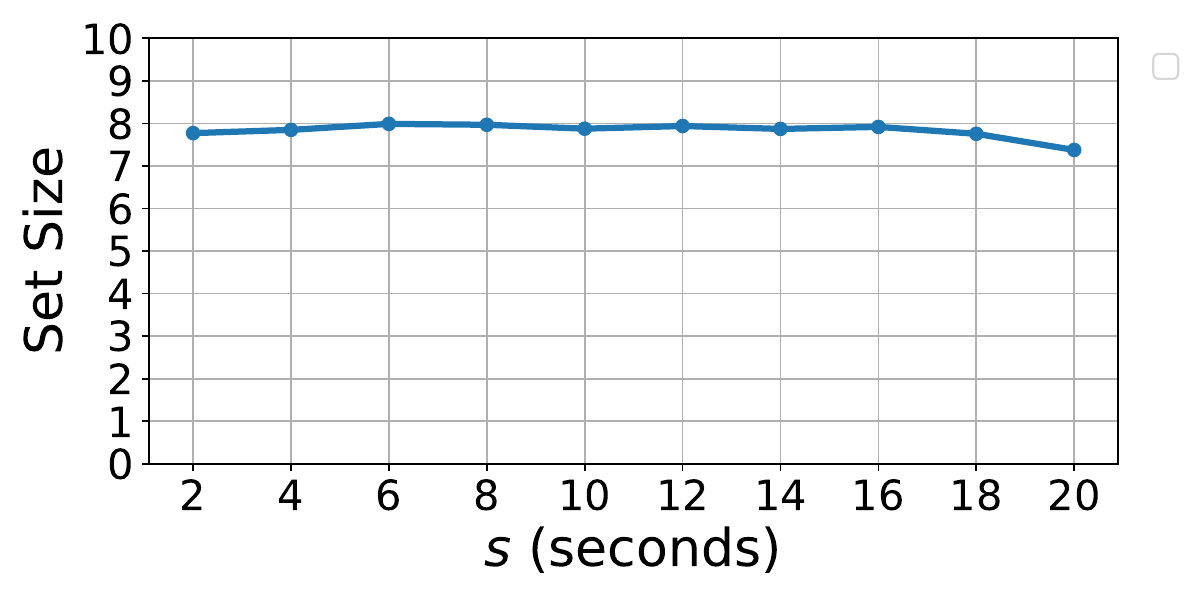}
\captionof{figure}{Average set size corresponding with Fig. \ref{fig:multiple}} 
\label{fig:multi_size}
\end{figure}

\end{document}